\documentclass[conference]{IEEEtran}
\usepackage[cmex10]{amsmath}
\usepackage{amssymb,amsfonts,amsthm,amscd,dsfont, mathrsfs,bbold,url,cite}
\usepackage{blkarray,verbatim,bm}
\usepackage{graphicx,float,psfrag,epsfig,color}

	\usepackage{microtype}
	\usepackage[pdftex,pagebackref=false,colorlinks]{hyperref}
	\hypersetup{linkcolor=[rgb]{.7,0,0}}
	\hypersetup{citecolor=[rgb]{0,.7,0}}
	\hypersetup{urlcolor=[rgb]{.7,0,.7}}

\newcommand{\Enote}[1]{\begin{center}\fbox{\begin{minipage}{20em}
                        {{\bf Emmanuel Note:} {#1}} \end{minipage}}\end{center}}
\newcommand{\Ynote}[1]{\begin{center}\fbox{\begin{minipage}{20em}
                        {{\bf Yuval Note:} {#1}} \end{minipage}}\end{center}}

\newcommand\independent{\protect\mathpalette{\protect\independent}{\perp}} 
\def\independent#1#2{\mathrel{\rlap{$#1#2$}\mkern2mu{#1#2}}}

\newcommand{\pg}{\mathrm{girth_*}}
\newcommand{\argmax}{\mathrm{argmax}}

\newcommand{\mR}{\mathbb{R}}

\newcommand{\pp}{\mathbb{P}}

\newcommand{\e}{\varepsilon}

\newcommand{\D}{\mathcal{D}}

\newtheorem{thm}{Theorem}
\newtheorem{lem}{Lemma}

\theoremstyle{definition}
\newtheorem{remark}{Remark}
\newtheorem{Def}{Definition}
\newtheorem{ex}{Example}
\DeclareMathOperator{\girth}{girth}
\DeclareMathOperator{\rank}{rank}
\newcommand{\pr}{\pp}
\newcommand{\E}{\mathbb{E}}
\newcommand{\F}{\mathbb{F}}

\begin{document}

\title{High-Girth Matrices and Polarization}
\author{\IEEEauthorblockN{Emmanuel Abbe}
\IEEEauthorblockA{Prog.\ in Applied and Computational Math.\ and EE Dept.\\
Princeton University\\
Email: \url{eabbe@princeton.edu}}
\and
\IEEEauthorblockN{Yuval Wigderson}
\IEEEauthorblockA{Department of Mathematics\\
Princeton University\\
Email: \url{yuvalw@princeton.edu}}
}
\date{}

\maketitle

\begin{abstract}
The girth of a matrix is the least number of linearly dependent columns, in contrast to the rank which is the largest number of linearly independent columns. This paper considers the construction of {\it high-girth} matrices, whose probabilistic girth is close to its rank. Random matrices can be used to show the existence of high-girth matrices with constant relative rank, but the construction is non-explicit. This paper uses a polar-like construction to obtain a deterministic and efficient construction of high-girth matrices for arbitrary fields and relative ranks. Applications to coding and sparse recovery are discussed. 
\end{abstract}

\section{Introduction}
Let $A$ be a matrix over a field $\F$. Assume that $A$ is flat, i.e., it has more columns than rows. The rank of $A$, denoted by $\rank(A)$, is the maximal number of linearly independent columns. The girth of $A$, denoted by $\girth(A)$, is the least number of linearly dependent columns. What are the possible tradeoffs between $\rank(A)$ and $\girth(A)$? This depends on the cardinality of the field. 
It is clear that 
\begin{align}
\girth(A) \leq \rank(A) + 1. \label{singleton}
\end{align}
Is it possible to have a {\it perfect-girth matrix} that achieves this upper-bound? If $\F=\mR$, drawing the matrix with i.i.d.\ standard Gaussian entries gives such an example with probability 1. However, if $\F=\F_q$, where $q$ is finite, the problem is different. 
For $\F=\F_q$, note that 
\begin{align}
\girth(A)= \mathrm{dist}(C_A), \label{g-d}
\end{align}
where $\mathrm{dist}(C_A)$ is the distance of the $q$-ary linear code $C_A$ whose parity check matrix is $A$. In fact, the least number of columns that are linearly dependent in $A$ is equal to the least number of columns whose linear combination can be made 0, which is equal to the least weight of a vector that is mapped to 0 by $A$, which is the least weight of a codeword, i.e., the code distance since the code is linear. 

Hence, over finite fields, the girth is a key parameter for error-correcting codes, and studying the girth/rank tradeoffs for matrices is equivalent to studying the distance/dimension tradeoffs for linear codes. Clearly it is not possible to obtain perfect-girth matrices over $\F=\F_2$, even if we relax the perfect-girth requirement to be asymptotic, requiring $\rank(A) \sim \girth(A)$ when the number of columns in $A$ tends to infinity.\footnote{We use the notation $a_n \sim b_n$ for $\lim_{n \to \infty} a_n/b_n =1$.}  If $\F=\F_2$, the Gilbert-Varshamov bound provides a lower-bound on the maximal girth (conjectured to be tight by some). Namely, for a uniformly drawn matrix $A$ with $n$ columns, with high probability,  
\begin{align}
\rank(A) = n H(\girth(A)/n) + o(n),
\end{align}
where $H$ is the binary entropy function. 

For $\F=\F_q$, the bound in \eqref{singleton} is a restatement of the Singleton bound for linear codes and expressed in terms of the co-dimension of the code. Asking for a perfect-girth matrix is hence equivalent to asking for an MDS linear code. Such constructions are known when $q=n$ with Reed-Solomon codes. Note that the interest on MDS codes has recently resurged with the applications in distributed data storage, see \cite{dimakis} for a survey. 

One may consider instead the case of non-finite fields, typically not covered in coding theory. As shown in Section \ref{sparse-sec}, this is relevant for the recovery of sparse signals \cite{donoho} via compressed measurements. The girth is then sometimes called differently, such as the Kruskal-rank or spark \cite{donoho}. As stated above, for $\F=\mR$, a random Gaussian matrix is perfect-girth with probability one. However, computing the girth of an arbitrary matrix is NP-hard \cite{tillman} (like computing the distance of a code \cite{vardy}), making the latter construction non-explicit.

In this paper, we are mainly interested in the following notion of {\it probabilistic girth}, defined to be the least number of columns that are linearly dependent with high probability, when drawing the columns uniformly at random. Formal definitions are given in the next section. Going from a worst-case to a probabilistic model naturally allows for much better bounds. In particular, defining {\it high-girth matrices} as matrices whose probabilistic girth and rank are of the same order (up to $o(n)$), a random uniform matrix proves the existence of high-girth matrices even for $\F=\F_2$. However, obtaining an explicit construction is again non-trivial. 

In this paper, we obtain explicit and fully deterministic constructions of high-girth matrices over any fields and for any relative ranks. We rely on a polar-code-like construction. Starting with the same squared matrix as for polar or Reed-Muller codes, i.e., the tensor-product/Sierpinski matrix, we then select rows with a different measure based on ranks. For finite fields, we show that high-girth matrices are equivalent to capacity-achieving linear codes for erasure channels, while for errors the speed of convergence of the probabilistic girth requirement matters. In particular, we achieve the Bhattacharyya bound with our explicit construction. For the real field, this allows to construct explicit binary measurement matrices with optimal probabilistic girth. 

These results have various other implications. First, our construction gives an operational interpretation to the upper-bound of the Bhattacharyya-process in polar codes. When the channel is not the BEC, the upper-bound of this process used in the polar code literature is in fact the conditional rank process studied in this paper. 
Second, this paper gives a high-school level proof (not necessarily trivial but relying only basic linear algebra concepts) of a fully deterministic, efficient, and capacity-achieving code for erasure channels.  While capacity-achieving codes for the BEC are well-known by now, most constructions rely still on rather sophisticated tools (expander codes, polar codes, LDPC codes, spatially-coupled codes), and we felt that an explicit construction relying only on the notion of rank and girth is rather interesting. On the other hand, for $\F=\F_2$, our construction turns out to be equivalent to the polar code for the BEC, so that the difference is mainly about the approach. It allows however to simplify the concepts, not requiring even the notion of mutual information.  Finally, we expect the result to generalize to non-binary alphabets, given that our construction does depend on the underlying field. 

\section{High-Girth Matrices}
\subsection{Notation}
Let $A$ be a $m \times n$ matrix over a field $\F$. For any set $S \subseteq [n]$, let $A[S]$ be the submatrix of $A$ obtained by selecting the columns of $A$ indexed by $S$. For $s \in [0,1]$, let $A[s]$ be a random submatrix of $A$ obtained by sampling each column independently with probability $s$. Thus, $A[s]=A[{\widetilde S}]$, where $\widetilde S$ is an i.i.d.\ Ber$(s)$ random subset of $[n]$. In expectation, $A[s]$ has $sn$ columns. Throughout the paper, an event $E_n$ takes place with high probability if $\pp\{E_n\} \to 1$ when $n \to \infty$, where $n$ should be clear from the context. 

\subsection{Probabilistic Girth}

\begin{Def}
Let $\{A_n\}$ be a sequence of matrices over a field $\F$, where $A_n$ has $n$ columns.
The {\it probabilistic girth} of $A_n$ is the supremum of all $s \in [0,1]$ such that $A_n[s]$ has linearly independent columns with high probability, i.e., 
\begin{align}
&\pg(\{A_n\}) :=\\\notag
&\sup \{s \in [0,1]: \pp \{A_n[s]\text{ has lin. indep. cols.}\}=1-o(1) \}
\end{align}
\end{Def}
Note that a better name would have been the probabilistic {\it relative} girth, since it is a counterpart of the usual notion of girth in the probabilistic setting with in addition a normalization factor by $n$. We often write $\pg(A_n)$ instead of $\pg(\{A_n\})$. We will sometimes care about how fast the above probability tends to 1. We then say that $A_n$ has a probabilistic girth with rate $\tau(n)$ if the above definition holds when
\begin{align}
\pp\{ A_n[s] \text{ has lin. indep. columns} \} = 1 - \tau(n).
\end{align}


\begin{Def} 
We say that $A_n$ is {\it high-girth} if 
\begin{align}
\pg(A_n) = \limsup_{n \to \infty} \rank(A_n)/n.
\end{align} 
For $\mu \in [0,1]$, we say that $A_n$ is $\mu$-high-girth if it is high-girth and $\pg(A_n)=\mu$. 
\end{Def}

\begin{ex} Consider the following construction, corresponding to Reed-Solomon codes. 
Let $x_1,\ldots,x_n$ be distinct elements of a field $\F$, and consider the $m \times n$ matrix
\begin{align}
V=\begin{pmatrix}
1&1&1&\cdots&1\\
x_1&x_2&x_3&\cdots&x_n\\
x_1^2&x_2^2&x_3^2&\cdots&x_n^2\\
\vdots&\vdots&\vdots&\ddots&\vdots\\
x_1^{m-1}&x_2^{m-1}&x_3^{m-1}&\cdots&x_n^{m-1}
\end{pmatrix}
\end{align}
Then $V$ will satisfy a stronger property than being high-girth, as its actual girth is $m+1$: \emph{every} $m \times m$ submatrix will be invertible, since every $m \times m$ submatrix is a Vandermonde matrix whose determinant must be nonzero. 
However, this example cannot be used to construct high-girth families over a fixed finite field $\F$. For as soon as $n$ exceeds the size of $\F$, it will be impossible to pick distinct $x_i$'s, and we will no longer have high girth.
\end{ex}

\begin{ex}
A $\mu n \times n$ uniform random matrix with entries in $\F_2$ is $\mu$-high-girth with high probability. 
\end{ex}

\section{Explicit construction of high-girth matrices}
\subsection{Sierpinski matrices}
Let $\F$ be any field, let $n$ be a power of $2$, and let $G_n$ be the matrix over $\F$ defined by
\begin{align}
G_n=\begin{pmatrix}
1&1\\0&1
\end{pmatrix}^{\otimes \log n}.
\end{align}
Note that the entries of this matrix are only $0$'s and $1$'s, hence this can be viewed as a matrix over any field. 

Many important codes can be derived from $G_n$, and they are all based on a simple idea. Namely, we first pick some measure of ``goodness'' on the rows of $G_n$. Then, we take the submatrix of $G_n$ obtained by keeping only those rows which are the ``best'' under this metric, and we finally define a code whose PCM is this matrix. The 
first important examples are Reed-Muller (RM) codes \cite{reed,muller}, where goodness is measured by the weight of the rows, and more recently polar codes \cite{arichannel,arisource}, where goodness is measured by the entropy (or mutual information). In the next section, we define a measure of goodness based on ranks and use it to construct high-girth matrices. A similar construction was proposed in \cite{renyipol} for Hadamard matrices to polarize the R\'enyi information dimension. We discuss applications to coding and sparse recovery in the next sections. 

\subsection{Conditional-rank matrices}
With $s \in [0,1]$ fixed, let $G_n^{(i)}$ denote the submatrix of $G_n$ obtained by taking the first $i$ rows, and let $G_n^{(i)}[s]$ be the random submatrix obtained by sampling each column independently with probability $s$, as above. 
\begin{Def}
The conditional rank (COR) of row $i$ in $G_n$ is defined by
\begin{align}
\rho(n,i,s)=\E(\rank_\F G_n^{(i)}[s])-\E(\rank_\F G_n^{(i-1)}[s])
\end{align}
where $\rank_\F$ denotes the rank computed over the field $\F$. When $i=1$, define
\begin{align}
\rho(n,i,s)=\E(\rank_\F G_n^{(1)}[s])
\end{align}
\end{Def}


Now, by adding the $i$th row, we will either keep the rank constant or increase it by $1$, and the latter will happen if and only if the $i$th row is independent of the previous rows. So we get that
\begin{align}
\rho(n,i,s)&=\pp(\text{the $i$th row of $G_n[s]$ is}\nonumber\\
&\text{independent of the previous $i-1$ rows}),
\end{align}
where linear independence is also considered over $\F$. The key property of the conditional ranks is expressed in the following lemma.

\begin{lem} \label{branching}
Define the functions
\begin{align}
\ell(x)&=2x-x^2\\
r(x)&=x^2
\end{align}
and define a branching process of depth $\log n$ and offspring 2 (i.e., each node has exactly two descendants) as follows: the base node has value $s$, and for a node with value $x$, its left-hand child has value $\ell(x)$ and its right-hand child has value $r(x)$. Then the $n$ leaf-nodes of this branching process are, in order, the values $\rho(n,i,s)$ for $1 \leq i \leq n$.
\end{lem}
An important point about this lemma is that the functions $\ell$ and $r$ do not depend on $\F$, while the $\rho(n,i,s)$ values do, a priori. Thus, one way to interpret this lemma is that the expected conditional ranks of $G_n$ do not depend on the field $\F$, even though their definition does. The proof of Lemma \ref{branching} is given in Section \ref{someproofs}.

A key property of the branching process in Lemma \ref{branching} is that it is a balanced process, meaning that the average value of the two children of a node with value $x$ is $x$ again:
\begin{align}
\frac{\ell(x)+r(x)}{2}=\frac{(2x-x^2)+x^2}{2}=x
\end{align}
This means that this branching process defines a martingale, by letting a random walk go left or right with probability half. Moreover, since $\rho(n,i,s)$ is a probability, we have that this martingale stays in $[0,1]$. So by Doob's martingale convergence theorem, we must have that this martingale converges almost surely to its fixed points. In fact, Doob's theorem is not needed here, as one may conclude using only the fact that the increments are orthogonal.\footnote{Private discussion with E.\ Telatar. See also \cite{apc529}}
 Its fixed points are those $x$'s for which $\ell(x)=r(x)=x$. The only points satisfying this are $0$ and $1$, so this martingale polarizes. 
In fact, much can be said about the speed of polarization of this process, as it is equivalent to the polarization process for BEC channels studied in \cite{arikantelatar}. 
\begin{thm}  [Application of \cite{arikantelatar}] \label{polspeed}
For any $n$,
\begin{align}
&\frac{|\{i \in [n]:\rho(n,i,s) >1-2^{-n^{0.49}}\}|}{n}=s+o(1)\\
&\frac{|\{i \in [n]:\rho(n,i,s)<2^{-n^{0.49}}\}|}{n}=(1-s)+o(1)
\end{align}
\end{thm}
Hence the theorem tells us is that the above martingale polarizes very quickly: apart from a vanishing fraction, all $\rho(n,i,s)$'s are exponentially close to $0$ or $1$ as $n \to \infty$. With this in mind, we define the following.
\begin{Def}
Let $n$ be a fixed power of $2$, and let $s \in [0,1]$ be fixed. Let $H \subset [n]$ be the set of indices $i$ for which $\rho(n,i,s) >1-2^{-n^{0.49}}$, and let $m=|H|$. By Theorem \ref{polspeed}, we know that $m=sn+o(n)$. Let $R_n$ denote the $m \times n$ submatrix of $G_n$ gotten by selecting all the columns of $G_n$, but only taking those rows indexed by $H$. We call $R_n$ the COR matrix of size $n$ with parameter $s$. 
\end{Def}

Note that the construction of COR matrices is trivial as opposed to the construction of polar codes based on the entropy of mutual information for general sources or channels.  

We will index the rows of $R_n$ by $i \in H$, rather than $j \in [m]$. We sometimes denote $R_n$ by $R$. The most important property of $R_n$ is expressed in the following theorem.

\begin{thm} \label{fullrank}
For any $s \in [0,1]$, $R_n[s]$ has full rank (i.e. rank $m$) with high probability, as $n \to \infty$. In fact, $R_n[s]$ has full rank with probability $1-o(2^{-n^{0.49}})$.
\end{thm}
The proof is a simple consequence of Lemma \ref{branching} and Theorem \ref{polspeed}, and can be found in the Appendix.

Theorem \ref{fullrank} implies the following.
\begin{thm}\label{rrgirth}
For any $s \in [0,1]$, $R_n$ is $s$-high-girth.
\end{thm}
Since the proof of Theorem \ref{fullrank} works independently of the base field $\F$, the same is true of Theorem \ref{rrgirth}. Thus, the COR construction is a fully deterministic and works over any field. In fact, it requires only two values (0 and 1) for the matrix, even when $\F=\mR$. 

\section{Applications of high-girth matrices}
\subsection{Coding for erasures}

Let $\F$ a field and $p \in [0,1]$. The memoryless erasure channel on $\F$ with erasure probability $p$, denoted by $MEC(p)$, erases each component of a codeword on $\F$ independently with probability $p$. Denoting by $\varepsilon$ the erasure symbol, the output alphabet is hence $\F_*=\F \cup \{\varepsilon\}$ and the transition probability of receiving $y \in \F_*$ when $x\in \F$ is transmitted is 
\begin{align}
W(y|x)=\begin{cases}
p& \text{if }y=\varepsilon,\\
1-p& \text{if } y=x.
\end{cases} 
\end{align}
The memoryless extension is defined by $W^n(y^n|x^n)=\prod_{i=1}^n W(y_i|x_i)$ for $x^n\in \F^n, y^n \in \F_*^n$.


Recall that a code of block length $n$ and dimension $k$ over the alphabet $\F$ is a subset of $\F^n$ of cardinality $|\F|^k$. The code is linear if the subset is a subspace of dimension $k$. In particular, a linear code can be expressed as the image of a generator matrix $G \in \F^{n \times k}$ or as the null space of a parity-check matrix $H \in \F^{(n-k) \times n}$. The rate of a code is defined by $k/n$. A rate $R$ is achievable over the $MEC(p)$ if the code can correct the erasures with high probability. More specifically, $R$ is achievable if there exists a sequence of codes $C_n$ of blocklength $n$ and dimension $k_n$ having rate $R$, and decoders $D_n: \F_*^n \to \F^n$, such that
$P_e(C_n)\to 0$, where for $x^n$ drawn uniformly at random in $C_n$ and $y^n$ the output of $x^n$ over the $MEC(p)$, and
\begin{align}
P_e(C_n) &:= \pp \{ D(y^n) \neq x^n \}.
\end{align} 
The dependency in $D_n$ is not explicitly stated in $P_e$ as there is no degree of freedom to decode over the MEC (besides guessing the erasure symbol), as shown in the proof of next lemma. 

The supremum of the achievable rates is the capacity, given by $1-p$. We now relate capacity-achieving codes on the MEC and high-girth matrices. 

\begin{lem}\label{cancorrecterasures}
A linear code $C_n$ achieves a rate $R$ on the $MEC(p)$ if and only if its parity check matrix has probabilistic girth at least $1-R$. In particular, a code achieves capacity on on the $MEC(p)$ if and only if its parity check matrix is $p$-high-girth.
\end{lem}
In particular, the linear code whose parity-check matrix is a COR matrix of parameter $p$ achieves capacity on the $MEC(p)$. 

\begin{remark}
In the binary case, COR codes give a new interpretation to BEC polar codes: instead of computing the mutual information of the polarized channels via the generator matrix, we can interpret BEC polar codes from the girth of the parity-check matrix. Note that this simplifies the proof that BEC polar codes achieve capacity to a high-school linear algebra --- mutual information need not be even introduced. 
The only part which may not be of a high-school level is the martingale argument, which is in fact not necessary, as already known in the polar code literature (see for example \cite[Homework~4]{apc529}, which basic algebra).
\end{remark}
\begin{remark}
As shown in \cite{arichannel}, the action of the matrix $G_n=\left (\begin{smallmatrix}1&1\\0&1\end{smallmatrix}\right)^{\otimes \log n}$ on a vector can be computed in $O(n \log n)$ time as well, which means that the encoding of the COR code can be done in $O(n\log n)$ time, as well as the code construction (which is not the case for general polar codes). Decoding the COR code can be done by inverting the submatrix of $A$ corresponding to the indices that do not have erasure symbols, which can be done by Gaussian elimination in $O(n^3)$ time. Alternatively, COR codes can be decoded as polar codes, i.e., successively, in $O(n\log(n))$. Hence, like polar codes for the BEC, COR codes are deterministic, capacity-achieving, and efficiently encodable and decodable for the MEC.
\end{remark}

\subsection{Sparse recovery}\label{sparse-sec}
In the setting of sparse recovery, one wishes to recover a real-valued sparse signal from a lower-dimensional projection \cite{donoho}. In the worst-case model, a $k$-sparse signal is a vector with at most $k$ non-zero components, and to recover $x$ that is $k$-sparse from $Ax$, it must be that $Ax \neq Ax'$ from all $x,x'$ that are $k$-sparse (and different). Hence $A$ needs\footnote{Note that for noise stability or to obtain a convex relaxation of the decoder, one needs the columns to have in addition singular values close to 1, i.e., the restricted isometry property (RIP).} to have girth $2k+1$.

One may instead consider a probabilistic model where a $k$-sparse signal has a random support, drawn uniformly at random or from an i.i.d.\ model where each component in $[n]$ belongs to the support with probability $p=k/n$. The goal is then to construct a flat matrix $A$ which allows to recover $k$-sparse signals with high probability on the drawing of the support. Note that a bad support $S$ is one which is $k$-sparse and that can be paired with another $k$-sparse support $S'$ such that that there exists real-valued vectors $x,x'$ supported respectively on $S,S'$ which have the same image through $A$, i.e.,
\begin{align}
Ax = Ax' \quad \Longleftrightarrow \quad A(x-x')=0.
\end{align}
Note now that this is equivalent to saying that the columns of $A$ indexed by $S \cup S'$ are linearly dependent, since $x-x'$ is supported on $S \cup S'$ which is $2k$-sparse. 

Hence, the probability of error for sparse recovery is given by 
\begin{align}
\pp_S\{\exists S' : A[S \cup S'] \text{ has lin.\ dep. columns}\}.
\end{align}
This error probability can be upper-bounded as for errors (see next section and Section \ref{new-sec}), by estimating the probability that $A$ has a subset of up to $2k$ linearly dependent columns, which relies on the high-girth property of $A$.

\subsection{Coding for errors}
 
In this section, we work over the binary field $\F_2$. The binary symmetric channel with error probability $p$, denoted by $BSC(p)$, flips each bit independently with probability $p$. More formally, the transition probability of receiving $y \in \F_2$ when $x \in \F_2$ is transmitted is given by
\begin{align}
W(y | x)=\begin{cases}
p&\text{if }y \neq x\\
1-p&\text{if }y=x
\end{cases}
\end{align}
The memoryless extension is then defined by $W^n(y^n | x^n)=\prod_{i=1}^n W(y_i | x_i)$, for $x^n,y^n \in \F_2^n$. 
\begin{thm}\label{cancorrecterrors}
Let $p \in [0,1/2]$ and $s=s(p)=2\sqrt{p(1-p)}$ be the Bhattacharyya parameter of the $BSC(p)$. Let $\{C_n\}$ be the COR code with parameter $s$ (the code whose PCM is $R_n$). Then $C_n$ can reliably communicate over the $BSC(p)$ with high probability, as $n \to \infty$. 
\end{thm}
Note that unlike in the erasure scenario, Theorem \ref{cancorrecterrors} does not allow us to achieve capacity over the $BSC(p)$. For the capacity of the $BSC(p)$ is $1-H(p)$, and
\begin{align}
1-H(p) \geq 1- 2\sqrt{p(1-p)}
\end{align}
with equality holding only for $p \in \{0,\frac12,1\}$. 

This statement, unlike the ones for erasure correction and sparse recovery, was stated only for the COR code, and not for general high-girth codes. Our proof of this statement, which can be found in the Appendix, relies on the actual construction of COR codes and on the upper-bound on the successive probability of error in terms of the COR known from polar codes. One may also attempt to obtain this result solely from the high-girth property, but this requires further dependencies on the high-girth rate of convergence (see Section \ref{new-sec}). It is an interesting problem to obtain achievable rates that solely depend on the probabilistic girth for the BSC. 

\section{Some proofs} \label{someproofs}
\begin{proof}[Proof of Lemma \ref{branching}]
We induct on $\log n$. The base case is $\log n=1$, where the calculation is straightforward. The rank of $G_2^{(1)}[s]$ will be $0$ if no columns are chosen, and will be $1$ if at least $1$ column is chosen. Therefore,
\begin{align}
\rho(2,1,s)&=\E(\rank_\F(G_2^{(1)}[s]))\\
&=0 \cdot (1-s)^2+1 \cdot 2s(1-s) +1 \cdot s^2\\
&=2s-s^2=\ell(s)
\end{align}
Similarly,
$
\E(\rank_\F(G_2^{(2)}[s]))=2s
$,
and thus
\begin{align}
\rho(2,2,s)=(2s)-(2s-s^2)=s^2=r(x)
\end{align}
Note that all these calculations do not actually depend on $\F$.

For the inductive step, assume that $\rho(n/2,i,s)$ is the leaf value of the branching process for all $1 \leq i \leq n/2$. To prove the same for $\rho(n,i,s)$, write 
$$
G_n=G_{n/2} \otimes G_2.
$$ In other words, we think of $G_n$ as being an $(n/2) \times (n/2)$ matrix whose entries are $2 \times 2$ matrices. 

We begin with the case when $i$ is odd. By the inductive hypothesis, we wish to prove that
\begin{align}
\rho(n,i,s)=\ell\left (\rho\left(\frac n2,\frac{i+1}{2},s\right )\right )
\end{align}
We partition the columns of $G_n^{(i)}[s]$ into two sets: $O$, which consists of those columns which have an odd index in $G_n$, and $E$, which consists of those with an even index in $G_n$. Since $G_n=G_{n/2} \otimes G_2$, we see that for $i$ odd, the $i$th row of $G_n$ is the $((i+1)/2)$th row of $G_{n/2}$, except that each entry is repeated twice. From this, and from inclusion-exclusion, we see that
\begin{align}
\rho(n,i,s)&=\pp(\text{row $i$ of $G_n$ is independent}\nonumber\\
&\hspace{1cm}\text{of the previous rows})\\
&=\pp(\text{row $i$ of $G_n[O]$ is independent}\nonumber\\
&\hspace{1cm}\text{of the previous rows of $G_n[O]$})\nonumber\\
&{}+\pp(\text{row $i$ of $G_n[E]$ is independent}\nonumber\\
&\hspace{1cm}\text{of the previous rows of $G_n[E]$})\nonumber\\
&{}-\pp(\text{both of the above})\\
&=2\rho\left(\frac n2,\frac{i+1}{2},s\right)-\rho\left(\frac n2,\frac{i+1}{2},s\right)^2\\
&=\ell\left (\rho\left(\frac n2,\frac{i+1}{2},s\right )\right )
\end{align}
Next, we consider the case when $i$ is even. In this case, we wish to prove that
\begin{align}
\rho(n,i,s)=r\left(\rho\left(\frac n2,\frac i2,s\right)\right)
\end{align}
We proceed analogously. From the equation $G_n=G_{n/2} \otimes G_2$, we see that the $i$th row of $G_n$ is the $(i/2)$th row of $G_{n/2}$, except with a $0$ intersprersed between every two entries. Thus, the $i$th row will be \emph{dependent} if either it restriced to $O$ or it restricted to $E$ will be dependent; in other words, it'll be independent if and only if both the restriction to $O$ and the restriction to $E$ are independent. Therefore,
\begin{align}
\rho(n,i,s)&=\pp(\text{the restriction to $O$ and}\nonumber\\
&\hspace{.3cm}\text{the restriction to $E$ are independent})\\
&=\pp(\text{the restriction to $O$ is independent}) \cdot\nonumber\\ &{}\quad \pr(\text{the restriction to $E$ is independent})\\
&=\left(\rho\left(\frac n2,\frac i2,s\right)\right)^2\\
&=r\left(\rho\left (\frac n2,\frac i2,s\right)\right)
\end{align}
Note that in all our calculations, we used probability arguments that are valid over any field. Broadly speaking, this works because the above arguments show that the only sorts of linear dependence that can be found in $G_n^{(i)}[s]$ involves coefficients in $\{-1,0,1\}$. Since these elements are found in any field, we have that this theorem is true for all fields $\F$.
\end{proof}

\begin{proof}[Proof of Lemma \ref{cancorrecterasures}]
Note that a decoder over the MEC needs to correct the erasures, but there is no bias towards which symbol can have been erased. Hence, a decoder on the MEC is wrong with probability at least half if there are multiple codewords that match the corrupted word. In other words, the probability of error is given by\footnote{If ties are broken at random, an additional factor of $1-1/|\F|$ should appear on the right hand side.} 
\begin{align}
P_e(C_n)&=\pp_E \{ \exists x,y \in C_n, x \neq y, x[E^c]=y[E^c]  \}  \\
&=\pp_E \{ \exists z \in C_n, z \neq 0, z[E^c]=0  \}  \label{e1}
\end{align} 
where $E$ is the erasure pattern of the $MEC(p)$, i.e., a random subset of $[n]$ obtained by picking each element with probability $p$. 
Let $H_n$ be the parity-check matrix of $C_n$, i.e., $C_n=\ker(H_n)$. 
Note that $E$ has the property that there exists a codeword $z \in C_n$ such that $z[E^c]=0$ if and only if the columns indexed by $E$ in $H_n$ are linearly dependent. Indeed, assume first that there exists such a codeword $z$, where the support of $z$ is contained in $E$. Since $z$ is in the kernel of $H_n$, the columns of $H_n$ indexed by the support of $z$ must add up to $0$, hence any set of columns that contains the support of $z$ must be linearly dependent. Conversely, if the columns of $H_n$ indexed by $E$ are linearly dependent, then there exists a subset of these columns and a collection of coefficients in $\F$ such that this linear combination is $0$, which defines the support of a codeword $z$. Hence,
\begin{align}
P_e(C_n)&=\pp_E \{ H_n[E] \text{ has lin.\ dependent columns} \} .\label{e2}
\end{align} 
Recalling that the code rate is given by $1-r$, where $r$ is the relative rank of the parity-check matrix, the conclusions follow.
\end{proof}

\section*{Acknowledgements}
This work was partially supported by NSF grant CIF-1706648 and the Princeton PACM Director's Fellowship.
\bibliographystyle{plain}
\bibliography{girthcites}

\section{Appendix}
\subsection{More Proofs}
\begin{proof}[Proof of Theorem \ref{fullrank}]
For $i \in H$, let $B_i$ be the event that the $i$th row of $R[s]$ is linearly dependent on the previous rows. Note that if $R[s]$ has full rank, then no $B_i$ is satisfied, while if $R[s]$ has non-full rank, then there must be some linear dependence in the rows, so at least one $B_i$ will be satisfied. In other words, the event whose probability we want to calculate is simply the event $\bigcap_{i\in H} B_i^c$. 


Note that in our notation, the $i$th row of $R$ is also the $i$th row of $G_n$. Therefore, for any $S \subseteq [n]$, the $i$th row of $R[S]$ is the $i$th row of $G_n[S]$. This means that any linear dependence between the $i$th row of $A[S]$ and the previous rows automatically induces a linear dependence between the $i$th row of $G_n[S]$ and the previous $i-1$ rows, since the previous rows in $G_n[S]$ are a superset of the previous rows in $R[S]$. Since this is true for any set $S \subseteq [n]$, we see that
\begin{align}
\pp(B_i)&=\pp(\text{the }i\text{th row of }R[s]\text{ is dependent on}\nonumber\\
&\hspace{1cm}\text{the previous rows of }R[s])\\
&\leq \pp(\text{the }i\text{th row of }G_n[s]\text{ is dependent on}\nonumber\\
&\hspace{1cm}\text{the previous rows of }G_n[s])\\
&=1-\rho(n,i,s)\\
&<2^{-n^{0.49}}
\end{align}
Therefore,
\begin{align}
\pr\left(\bigcap_{i\in H} B_i^c\right)&=1-\pr\left(\bigcup_{i\in H} B_i\right)\\
\hspace{1cm}&\geq 1-\sum_{i\in H} \pr(B_i)\\
&> 1-\sum_{i=1}^m 2^{-n^{0.49}}\\
&=1-[sn+o(n)]2^{-n^{0.49}}\\
&=1-o\left(2^{-n^{0.49}}\right)\\
&\to 1 \text{ as }n \to \infty
\end{align}

\end{proof}

\begin{proof}[Proof of Theorem \ref{cancorrecterrors}]
We recall from \cite{arichannel} that in any code generated by taking some of the rows of $G_n$, we have that the probability of error on the $BSC(p)$ is upper-bounded as
\begin{align}\label{bhattbound}
P_e \leq \sum_{i \in H^c} Z_n^i
\end{align}
where $H$ denotes the set of rows of $G_n$ that we keep and $Z_n^i$ denotes the Bhattacharyya parameter of the $i$th row. Proposition 5 in \cite{arichannel} tells us that 
\begin{align}
Z_n^i\begin{cases}
=(Z_{n/2}^{i/2})^2&\text{ when $i$ is even}\\
\leq 2(Z_{n/2}^{(i+1)/2})-(Z_{n/2}^{(i+1)/2})^2&\text{ when $i$ is odd}
\end{cases}
\end{align}
We recognize these functions as $r$ and $\ell$. Thus, we see that the branching process of Lemma \ref{branching}, when initialized at $s=Z(BSC(p))$, provides an upper bound for the Bhattacharyya parameters. Now, recall that the row selection criterion for COR matrices only keeps the rows with high $\rho$ values, and thus high $Z$ values. Thus, we see that (\ref{bhattbound}) ensures that COR codes with parameter $s=2\sqrt{p(1-p)}$ can successfully transmit over the $BSC(p)$.
\end{proof}

\section{Errors from high-girth matrices}\label{new-sec}
It is an interesting to study what rate on the high-girth property allows to achieve positive rates on the BSC. A classic union-bound requires at least exponential rate, as explained next. This underlines that COR matrices achieve rates higher than what arbitrary high-girth matrices may reach.

Let $H=H_n$ be a matrix with probabilistic girth $\mu$. Consider $C$ to be the code whose parity check matrix is $H$. The probability of error of this code on the BSC$(p)$ is the probability that an error vector $Z$ has in its coset (i.e., the other error vectors that lead to the same syndrome) a more likely vector, i.e., 
\begin{align}
P_e = \pp_Z\{\exists z' \in \F_2^n: HZ=Hz', w(z') \leq w(Z)\}
\end{align}
where $Z$ is i.i.d.\ Bernoulli$(p)$. 
This is equivalent to 
\begin{align}
P_e = \pp_Z\{\exists x \in \F_2^n: Hx=0, w(x+Z) \leq w(Z)\}.
\end{align}
Note that $w(x+Z) \leq w(Z)$ means that $Z$ takes more often the value 1 in the support of $x$, 
which has probability 
\begin{align}
\sum_{k=w/2}^w {\binom wk} p^k (1-p)^{w-k},
\end{align}
where $w/2$ is rounded up if not even. Note that 
\begin{align}
&\sum_{k=w/2}^w {\binom wk} p^k (1-p)^{w-k} \\
&\leq p^{w/2} (1-p)^{w/2} 2^w = z(p)^w,
\end{align}
where 
\begin{align}
z(p) = 2p^{1/2} (1-p)^{1/2}
\end{align}
is the Bhattacharyya parameter of the channel. 
Hence,
\begin{align}
P_e &\leq \sum_{w\geq 1} N(w) z(p)^w
\end{align}
where $N(w)$ is the number of codewords of weight $w$, i.e.,
\begin{align}
N(w)=|\{x: Hx=0, w(x)=w\}|.
\end{align}
Note that 
\begin{align}
N(w)&\leq |\{x: H[x] \text{ has lin.\ dep.\ col.}, w(x)=w\}|.\\
&= {\binom nw} \pp_S\{ H[S] \text{ has lin.\ dep.\ col.}\},
\end{align}
where $S$ is uniformly drawn with a support of size $w$. By standard concentration arguments,
this is upper-bounded by $\pp\{ H[s] \text{ has lin.\ dep.\ col.}\}$ where $s=w/n + o(\sqrt{w/n})$, and 
by assumption, 
\begin{align}
\pp\{ H[s] \text{ has lin.\ dep.\ col.}\}=\tau(n)
\end{align}
when $s < \mu$. Thus, 
\begin{align}
P_e &\leq  \sum_{t < \mu} \tau(n) 2^{n(H(t)+t \log(z(p)))} + \sum_{t \geq \mu} N(tn)  z(p)^{tn} 
\end{align}
where $t$ takes values such that $tn$ is an integer. Hence $\tau(n)$ must be exponentially small to drive the first term to 0.

\end{document}